\PassOptionsToPackage{hyphens}{url}
\documentclass[a4paper,USenglish]{lipics-v2019}
\hideLIPIcs
\nolinenumbers

\usepackage{amsmath}
\usepackage{xspace}

\usepackage{booktabs}

\usepackage{natbib}

\newcommand{\fakepar}[1]{\medskip\par\textit{#1}\ }

\newcommand{\AR}{R\xspace}
\newcommand{\ABern}{B\xspace}

\newcommand{\Exp}[1]{{\rm E}[ #1 ]}
\newcommand{\Oh}[1]{\mathcal{O}\!\left( #1\right)}
\newcommand{\ceil}[1]{\left\lceil #1\right\rceil}
\newcommand{\floor}[1]{\left\lfloor #1\right\rfloor}
\newcommand{\prob}[1]{\text{Pr}\left[ #1 \right]}
\newcommand{\punkt}{\enspace .}
\newcommand{\setGilt}[2]{\left\{ #1: #2\right\}}
\newcommand{\set}[1]{\left\{ #1\right\}}

\newenvironment{code}{\noindent\it%
\begin{tabbing}%
\hspace{1.5em}\=\hspace{1.5em}\=\hspace{1.5em}\=\hspace{1.5em}\=\hspace{1.5em}\=%
\hspace{1.5em}\=\hspace{1.5em}\=\hspace{1.5em}\=\hspace{1.5em}\=\hspace{1.5em}\=%
\kill}{\end{tabbing}}
\newcommand{\Else}     {{\bf else\ }}
\newcommand{\Function} {{\bf Function\ }}
\newcommand{\Id}[1]{\ensuremath{\mathit{#1}}}
\newcommand{\If}       {{\bf if\ }}
\newcommand{\Is}{\mbox{\rm := }}
\newcommand{\RRem}[1]   {\`{\bf //\hspace{0.5mm}~}{\rm#1}}
\newcommand{\Return}   {{\bf return\ }}
\newcommand{\Then}     {{\bf then\ }}

\title{Efficient Parallel Random Sampling -- Vectorized, Cache-Efficient, and Online}
\titlerunning{Efficient Parallel Random Sampling -- Vectorized, Cache-Efficient, and Online}
\author{Peter Sanders}{Karlsruhe Institute of Technology}{sanders@kit.edu}{}{}
\author{Sebastian Lamm}{Karlsruhe Institute of Technology}{lamm@kit.edu}{}{}
\author{Lorenz Hübschle-Schneider}{Karlsruhe Institute of Technology}{huebschle@kit.edu}{}{}
\author{Emanuel Schrade}{Karlsruhe Institute of Technology}{schrade@kit.edu}{}{}
\author{Carsten Dachsbacher}{Karlsruhe Institute of Technology}{dachsbacher@kit.edu}{}{}

\authorrunning{Sanders, Lamm, Hübschle-Schneider, Schrade, Dachsbacher}
\Copyright{Peter Sanders, Sebastian Lamm, Lorenz Hübschle-Schneider, Emanuel Schrade, Carsten Dachsbacher}



\relatedversion{© ACM, 2018. This is the author's version of the work. It is posted here by permission of ACM for your personal use. Not for redistribution. The definitive version was published in ACM Transactions on Mathematical Software (TOMS), Volume 44, Issue 3 (April 2018), \url{https://doi.acm.org/10.1145/3157734}.}

\begin{document}

\maketitle

\begin{abstract}
We consider the problem of sampling $n$ numbers from the range $\{1,\ldots,N\}$
without replacement on modern architectures.  The main result is a
simple divide-and-conquer scheme that makes sequential algorithms more
cache efficient and leads to a parallel algorithm running in expected
time $\Oh{n/p+\log p}$ on $p$ processors, i.e., scales to massively parallel machines even for moderate values of $n$.
The amount of communication between the processors is
very small (at most $\Oh{\log p}$) and independent of the sample size.  We
also discuss modifications needed for load balancing,
online sampling, sampling with replacement,
Bernoulli sampling, and vectorization on SIMD units or GPUs.
\end{abstract}

\section{Introduction}

Random sampling is a fundamental ingredient in many algorithms, e.g.
for data analysis.  With the advent of ever larger data sets (``Big
Data''), the number of elements sampled from and even the sample
itself can become huge. Often the subsequent processing of the sample is
comparatively fast, and thus taking the sample can become a performance
bottleneck. Moreover, the speed of a single processor is stagnating so
that parallel algorithms are required for efficient sampling.
Furthermore, we can observe that only \emph{local} processing yields fast parallel algorithms and promises to scale linearly with the number of processors $p$.
Processor coordination over global memory or even
communication over the network quickly becomes a bottleneck~\citep{SSM13}. This is
particularly true for big data problems which often run on cloud
resources with limited communication capabilities, or for high
performance computing where the largest configurations are limited by
the bisection bandwidth of the network.

In this paper we focus on
the classical problem of sampling $n$ numbers out of the range $1..N$ without replacement.%
\footnote{We use the notation $a..b$ as shorthand for
  $\set{a,\ldots,b}$.}  In Section~\ref{s:preliminaries} we discuss
building blocks and previous approaches. Section~\ref{s:algorithm}
introduces our divide-and-conquer algorithm for sampling without
replacement. We discuss a number of generalizations in
Section~\ref{s:generalizations} including online sampling in sorted
order, load balancing, uneven distribution of the sampled universe,
using true randomness, achieving deterministic
results, sampling with replacement, Bernoulli sampling, and
vectorization. After providing details of our implementation in
Section~\ref{s:implementation}, Section~\ref{s:experiments} describes
experiments which demonstrate the speed and scalability of our
algorithm on both CPUs and GPUs. Section~\ref{s:conclusion} summarizes
the results and discusses some applications.

\section{Preliminaries and Related Work}\label{s:preliminaries}
\begin{table}[b]
  \centering
  \caption{Overview of abbreviations for the algorithms.\label{tab:algorithms}}
  \begin{tabular}{lllll}\toprule
    Abbrv. & Source                      & Time $\Oh{\cdot}$ & Space $\Oh{\cdot}$ & Mnemonic aid\\\midrule
    S      & \cite{Fan62}                & $N$               & $1$ & scan\\
    D      & \cite{Vit84}                & $n$               & $1$ & distance\\
    \ABern & \cite{AhrDie85}             & $n$               & $n$ & Bernoulli\\
    H      & \cite{Knu98}                & $n$               & $n$ & hash\\
    R      & Section~\ref{ss:sequential} & $n$               & $\log n$ & recursive\\
    P      & Section~\ref{ss:parallel}   & $n/p+\log p$      & $\log n$ & parallel\\
    SR     & Section~\ref{ss:sorted}     & $n$               & $1$ & sorted recursive\\\bottomrule
  \end{tabular}
\end{table}

Our goal is to efficiently take a sample of size $n$ from the range
$1..N$ using $p$ processors.  This algorithm can also be used to
sample from an array of $N$ elements.  More generally, the result
applies to sampling from a set $M$ of elements if $N=|M|$ is known and
if random access to the elements is possible.  Besides arrays, this
assumptions is true in many other situations, e.g., for files of equal
sized objects, or for tables in many database systems. In particular
this applies to column-oriented systems where random access is a basic
mechanism needed to reconstruct result rows after applying filters to
a small number of columns.  Many other database systems provide
random access. For example, the NoSQL system MongoDB by default
creates an index on the unique object-ids.

To avoid special case discussions, we will
henceforth assume that $n\leq N/2$.  For the unusual case $n>N/2$, one
can simply generate the $N-n<N/2$ elements that are \emph{not} in the
sample.
When considering parallel algorithms, we use $p$ to denote the number of
processing elements (PEs), which we assume to be connected by a network.  PEs
are numbered $1..p$.

\paragraph*{Algorithm S} \cite{Fan62} and \cite{Knu81} scan the range $1..N$ and generate a
uniformly distributed random deviate for each element to decide whether it is sampled.  For
$N\gg n$ this is prohibitively slow and we are surprised that the
algorithm still seems to be widely used, for example by the GNU
Scientific Library, GSL (function \texttt{gsl\_ran\_choose},
\url{https://www.gnu.org/software/gsl/}, version 2.2.1).

\paragraph*{Algorithm H}
Algorithm H is a simple and efficient folklore algorithm that is very good
for small~$n$ (see also \cite{AhrDie85}). The sample is kept in a hash table $T$ which is initially
empty. To produce the next sample element, it generates uniform deviates $X$ from
$1..N$. If $X\in T$, it rejects $X$, otherwise $X$ is inserted into $T$.
This algorithm runs in expected time
$\Oh{n}$. Note that $T$ contains random numbers and hence we can use a
very simple hash function, such as extracting the most significant $\log
n+\Oh{1}$ bits from the key.%
\footnote{In this paper $\log x$ stands for $\log_2 x$.}  For $n\ll
N$ the number of random deviates required is close to $n$, and we only
need uniform deviates. This makes Algorithm~H very fast for small
$n$.  For large~$n$, however, most hash table accesses cause cache faults,
slowing it down considerably.

Algorithm~H could be parallelized. However, the hash table accesses
then become global interactions between the PEs. The resulting
overheads are even larger than the cache faults in the sequential
algorithm and cause a severe bottleneck in distributed settings.  One
also has to be very careful if the resulting algorithm is supposed to
be \emph{deterministic}, i.e., the generated sample should be the same
in repeated runs with the same seeds for the random number generators:
race conditions in remote memory accesses or message delivery can
easily lead to differences in the generated sample.

\paragraph*{Algorithm~D}
\cite{Vit84} proposed an elegant sequential algorithm that
generates the samples in sorted order without any need for auxiliary
data structures. For generating the next sample, Algorithm~D
essentially generates an appropriate random deviate that specifies the
number of positions to skip.
Note that the random deviate changes in each step; using sophisticated
techniques based on the rejection method, generating these random deviates can be done in constant
expected time.

\paragraph*{Algorithm~\ABern}
\cite{AhrDie85} use the observation that taking a
Bernoulli sample where each element of $1..N$ is sampled with
probability $\rho\approx n/N$ yields a sample with $n'\approx n$
elements. If this sample is too big it can be repaired by removing
$n'-n$ of the elements randomly. By choosing $\rho$ somewhat larger
than $n/N$, one can make the case $n'<n$ highly unlikely and simply
restart the sampling process if it does occur.  Bernoulli sampling
can be implemented efficiently by generating geometrically distributed
random deviates to determine how many elements to skip in each
step. Algorithm~\ABern is faster than Algorithm~D because generating
geometric random deviates needs less arithmetic operations per element. In
Section~\ref{ss:Bernoulli} we point out that it may be even more
important that the parameters of the generated distribution remain the
same, as this makes vectorization possible. A notable difference
of Algorithm~\ABern to the aforementioned approaches is
that elements are not generated online, i.e., we have an initial delay of
$\Theta(n)$ before the first sample is generated.

\cite{meng2013scalable} gives a simple parallel sampling
algorithm that works well when the data is stored on disk.  However,
it only works with high probability, needs to compute a random
deviate for every input element, and needs communication between the
PEs in order to perform a parallel sorting step of elements that may
or may not be in the sample.

Table~\ref{tab:algorithms} summarizes the algorithm abbreviations used in this paper.

\section{Divide-and-Conquer Sampling}\label{s:algorithm}

Our central observation is that for any splitting position $\ell$, the
number of samples $L$ from the left range $1..\ell$ is distributed
hypergeometrically with parameters $n$ (number of experiments), $\ell$
(number of success states) and $N$ (universe size). Consequently, the number of samples from the right part
$\ell+1..N$ is $n-L$. We now apply this idea first to a sequential sampling algorithm in Section~\ref{ss:sequential} and then give a parallelization  in Section~\ref{ss:parallel}.

\subsection{Sequential Divide-and-Conquer Sampling (Algorithm~R)}\label{ss:sequential}

Algorithm~R in Figure~\ref{fig:algorithmR} gives pseudocode for a
sequential recursive divide-and-conquer algorithm based on the
splitting approach described above.  The range $1..N$ is split at $\ell=\floor{N/2}$.
Then $A$ is assigned to a recursively constructed sample of size $L$
from $1..\ell$.  $B$ becomes a sample of size $n-L$ from $\ell+1..N$.
After adding $\ell$ to the samples of $B$, we overall obtain the
desired sample of size $n$ from $1..N$.

The tuning parameter $n_0$ decides when to switch to the
base case. When using Algorithm~H, $n_0$ should be small enough so that
the hash table fits into cache. Note that the resulting recursion tree
has a size of at most $2n/n_0$. Hence the overall expected running time is
$\Oh{n}$, provided that we use a constant time algorithm for generating
hypergeometric random deviates (e.g. \cite{Stad90hyp}) and a linear
expected time algorithm for the base case.

\begin{figure}[h]
\begin{code}
\Function \Id{sampleR}$(n,N)$\+\\
  \If $n<n_0$ \Then \Return sampleBase$(n,N)$\RRem{e.g. using algorithms H or D}\\
  $L\Is \Id{hyperGeometricDeviate}(n,\floor{N/2}, N)$\\
  $A\Is \Id{sampleR}(L,\floor{N/2})$\\
  $B\Is \Id{sampleR}(n-L,N-\floor{N/2})$\\
  \Return $A\cup\setGilt{x+\floor{N/2}}{x\in B}$
\end{code}
\caption{\label{fig:algorithmR}Algorithm~\AR for (sequential) divide-and-conquer
  sampling without replacement.}
\end{figure}

\subsection{Parallel Divide-and-Conquer Sampling (Algorithm~P)}\label{ss:parallel}

We now describe a parallel sampling algorithm that needs almost no
communication. More precisely, each PE needs to know
$N$ and $n$ but otherwise, no communication is necessary.  If we can
assume that these values are known from the context, no communication
is needed at all. When these parameters are initially only known
at PE $1$, they can be broadcast with latency $\Oh{\log p}$.  In
subsequent refinements, slightly larger amounts of communication are
needed, but we will see that total communication overhead remains limited
to $\Oh{\log p}$.

For parallel sampling, we partition the range $1..N$ into $p$ pieces.
Let $N_i$ denote the last element in the range associated with PE
$i$, i.e., PE $i$ generates the sample elements that lie in the
subrange $N_{i-1}\!+\!1\,..\,N_i$ with $N_0\Is 0$. The underlying idea of the
parallelization is to adapt Algorithm~\AR in a way such that $\ceil{\log p}$
levels of recursion split the original range $1..N$ into the subranges
of each PE.
Initially, all PEs work on the full range $1..N$.
We use the PE numbers to gradually break this symmetry.
After splitting the subproblem, each PE will follow only a single recursive
call -- the one whose range contains its local subrange.

\begin{figure}
\begin{code}
\Function \Id{sampleP}$(n',j..k,i,h)$\+\\
  \If $k-j=1$ \Then\+\\
    use $h(i)$ to seed the local pseudorandom number generator\\
    $M\Is sampleLocally(n',N_i-N_{i-1}+1)$\RRem{e.g. using algorithms H, D, or \AR}\\
    \Return $\setGilt{N_{i-1}+x}{x\in M}$\-\\
  $m\Is\lfloor\frac{j+k}{2}\rfloor$\RRem{middle PE number}\\
  $L\Is \Id{hyperGeometricDeviate}(n',N_m-N_j+1, N_k-N_j+1,j..k,h)$\\
  \If $i\leq m$ \Then \Return $\Id{sampleP}(L,j..m,i,h)$\\
  \Else \Return $\Id{sampleP}(n'-L,m+1..k,i,h)$
\end{code}
\caption{\label{fig:algorithmP}Algorithm~P for sampling $n'$ elements
  on PEs $j..k$ where $i\in j..k$ is the PE executing the function.
  The initial call  on PE $i$ is
  $\Id{sampleP}(n,1..p,i,h)$.}
\end{figure}

Locally, when each PE works on its original range, we can use any sequential
algorithm, but we have to be careful: on the one hand, PEs following
the same path in this recursion tree have to generate the \emph{same}
random deviates to get a consistent result. On the other hand, random
deviates generated in two \emph{different} subtrees have to be
independent.  With true randomness (e.g. generated using a hardware
random number generator \citep{Intel12}) this would require
communication to distribute the right random values to the
PEs (see also Section~\ref{ss:true}). However, using pseudorandomness (as most applications do)
allows us to achieve the desired effect without any communication.
The idea is to use a (high quality) hash function $h$
as source of pseudorandomness for generating hypergeometric deviates.
In the subproblem for PEs $j..k$, the $t$-th random deviate is the hash of the triple $(j,k,t)$, i.e.
$h((j,k,t))$.  Figure~\ref{fig:algorithmP} gives pseudocode where the
function \Id{hyperGeometricDeviate} is passed both $h$ and $j..k$ in
order to be able to use this technique.  Within
function \Id{sampleLocally}, we can still use an ordinary generator of
pseudorandomness which may have a better trade-off between speed and
quality than hashing. In order to break the symmetry between the
PEs, we can seed it with~$h(i)$ on PE $i$.

Another issue is that the PEs need access to the global element
indices $N_j$. If the universe is evenly distributed between
PEs (except for the last one if $p$ does not divide $n$) this
is easy, as we simply have $N_j=j\ceil{n/p}$ for $j<p$ and $N_p=N$. Refer
to Section~\ref{sss:uneven} for the case of uneven
distribution of the universe.

\medskip\par\noindent
We obtain the following running time for Algorithm~P.
\begin{theorem}\label{thm:time}
If $\max_i\left(N_i-N_{i-1}\right)=\Oh{N/p}$ then Algorithm~P runs in time $\Oh{n/p+\log p}$
with high probability.%
\footnote{We say a statement is true with high probability if it is
  true with probability at least $1-p^{-c}$ for any constant $c$.}
\end{theorem}
\begin{proof}
First note that the bound holds when we only calculate with expectations.
Each PE generates $\leq\ceil{\log p}$ hypergeometric random deviates and $\Oh{n/p}$
samples in expectation.

However, we also have to take into account rare cases that could slow down
computation on some PE which would then lead to a
large overall execution time. Three issues have to be considered:
deviations in the number of samples per PE, deviations in the
time needed to generate the random deviates, and running time fluctuations within
function \Id{sampleLocally}.

The number of samples generated by one PE has a hypergeometric
distribution.  We exploit that this distribution spreads the elements
more evenly than a binomial distribution \citep[Theorem~3.3]{San96e} and analyze the
simpler situation when each sample is independently assigned to
PE $i$ with probability $(N_i-N_{i-1}+1)/N=\Oh{1/p}$. We thus
have a classical balls-into-bin situation that can be analyzed using
Chernoff bounds.
The particular calculations needed here have been done in \cite[Theorem~3.7]{San96e}
and yield exactly what we need --
$\Oh{\log p}$ samples with high probability when $n=\Oh{p\log p}$ and
$\Oh{n/p}$ samples with high probability when $n=\Omega(p\log p)$.

Fast algorithms for generating hypergeometric deviates \citep{Stad90hyp} are
often based on a rejection method, i.e., they generate a constant number
of uniform deviates, perform a constant amount of computation, and
then perform a test that succeeds with constant probability. If the
test fails, an independent new trial is performed. Hence, the running
time of the generator can be bounded by a constant times a geometrically distributed
random variable. Since each PE generates $\ceil{\log p}$ hypergeometric deviates, we need
a tail bound on the sum of a logarithmic number of exponential deviates with constant expectation.
By Lemma~\ref{lem:sum} below this sum is $\Oh{\log p}$
with high probability. By choosing the parameter $c$ in the definition of ``with high probability''
by one larger than what we want to show in the overall theorem, we can accommodate the fact that
the running time on the slowest PE determines the overall running time -- the probability that
any of the $p$ PEs is slower than claimed is bounded by $p\cdot p^{-(1+c)}=p^{-c}$.

When using Algorithm~D for generating local samples, we can use a
similar argument as above~-- the running time for generating each sample
is bounded by a geometrically distributed random variable so that
large deviations from the expectation are unlikely. When using
Algorithm~H, the details of the analysis depend on the tails of the
running time distribution of the hash table, but we will get the
required short tails of the running time distribution if we allocate
enough space~-- $\Oh{n/p+\log p}$ -- for this table. When using Algorithm~\AR,
additional hypergeometric random deviates are generated, but the
argument with the geometrically distributed running time again holds.
\end{proof}
\begin{lemma}\label{lem:sum}
  Consider $k=\Omega( \log p)$ independent geometrically distributed random variables $X_i$ with constant expectation.
  Then their sum $Y$ is in $\Oh{k}$ with high probability.
\end{lemma}
\begin{proof}
  By considering the geometric random variable with largest
  expectation, we can assume wlog that the variables are identically
  distributed. The sum of $k$ independent, identically distributed
  random variables has a negative binomial distribution. This situation can
  be analyzed using Chernoff bounds by exploiting the relation between binomial
  distribution and negative binomial distribution.
  We obtain:
  $$\prob{Y>a\Exp{Y}}\leq e^{-\frac{ak(1-1/a)^2}{2}}
  \leq p^{-\Omega(1)\frac{a(1-1/a)^2}{2}}
  \leq p^{-c}\punkt$$
  The first ``$\leq$'' follows a derivation of \cite{brown6wasted}.
  The second ``$\leq$'' exploits that $k=\Omega( \log p)$, and
  the third inequality chooses a sufficiently large constant $a$ for any particular choice of the
  constant $c$ from the definition of ``with high probability''.
\end{proof}

\section{Generalizations}\label{s:generalizations}

\subsection{Generating Output in Sorted Order and Online}\label{ss:sorted}
Note that Algorithm~\AR can easily output the elements in sorted
order provided that the base case algorithm generates the samples in
sorted order. This is certainly the case when using Algorithm~D, and we
can also adapt Algorithm~H for this purpose. For example, we can maintain the invariant
that the samples in the hash table are sorted. This is possible since
we use the most significant bits to address the table.  We only have
to ensure that colliding elements are also sorted. Rather than
appending an inserted element $k$ to the end of a cluster of colliding
elements, we skip elements smaller than $k$ and then shift the cluster
elements larger than $k$ one position to the right.  This makes
handling clusters of colliding elements somewhat slower, but the overall
overhead is small since the clusters are small (expected constant size).

Alternatively, we can insert into the hash table normally -- ignoring the
ordering of the keys -- and sort the hash table afterwards. Since the
sorting order is the same as the hash function value, the only thing
we have to do is to scan the hash table and sort clusters of colliding
elements. This leads to a linear time algorithm since the clusters are small.

It is also easy to modify Algorithm~R to generate samples online with
constant expected delay between generated samples. We can modify the
divide-and-conquer step to split off a range of size
$\ceil{N\cdot n_0/n}$. Using this splitting in an iterative fashion, we
generate samples in batches of expected size of approximately $n_0$. This takes
time $\Oh{n_0}$ per batch, i.e., constant time if $n_0$ is a constant.

The same techniques can be used in a parallel setting. Then each
PE generates the elements of its designated subrange of $1..N$
in sorted order.

\subsection{Load Balancing}\label{sss:lb}
Algorithm~\AR implicitly assumes that all PEs are equally
fast. However, for various reasons, this may not be the case. For
example, we might work with heterogeneous cloud resources, there might
be other jobs (or operating system services) slowing down some
PEs, or uneven cooling might imply different clock frequencies
for different PEs. In these cases, the slowest PE would
slow down the overall computation. This problem can be solved with
standard load balancing techniques.  We split~$1..N$ into $p'\gg p$
jobs (subranges) and use a load balancing
algorithm to dynamically assign jobs to PEs.

The most widely used load balancing method for such problems uses a
centralized master PE to assign jobs to
PEs. Unfortunately, this increases the running time from
$\Oh{n/p+\log p}$ (Theorem~\ref{thm:time}) to $\Oh{n/p+p}$. A more
scalable approach is \emph{work stealing} \citep{FinMan87,BluLei99}.
To employ this approach, we
instantiate the concept of a \emph{tree shaped computation}
\citep{San02b}: We conceptually split the work into very fine grained
\emph{atomic} jobs corresponding to ranges of sample values that are
expected to contain a constant number of samples (say,
$n_0$). However, initially these jobs are coalesced into $p$ (meta)
jobs of about equal size.  Now each PE sequentially works on
its meta job, one atomic job at a time. Idle PEs ask random
other PEs to split their range of unfinished atomic jobs in
half, delivering one half to the idle PE. Note that both
splitting off the next atomic job and splitting the remaining range of
atomic jobs in half can be done in constant expected time using the
division strategy from Algorithm~P. The generic analysis of
\cite{San02b} then yields the same asymptotic running time as in
Theorem~\ref{thm:time}.

At least on shared memory machines, this asymptotically scalable load
balancing is easy to implement since it is part of widely used tools
such as the C++ standard library \citep{SSP07}, the Intel Thread
Building Blocks, or Cilk \citep{BluEtAl95}.

\subsection{Uneven Distribution of the Sampled Universe}\label{sss:uneven}
When we sample from a set of elements distributed over PEs
connected by a network, we may not want to load balance. Rather, we want
to use the \emph{owner computes} paradigm -- each PE computes
those samples that stem from its local subset of elements.%
\footnote{In a hybrid setting, where several shared memory machines
  are connected by a network, we could still apply load balancing on each
  shared memory machine.}  In this situation, each PE
$i$ initially only knows its local number of elements $L_i$.

\begin{figure}[bt]
  \centering
  \includegraphics[width=\textwidth]{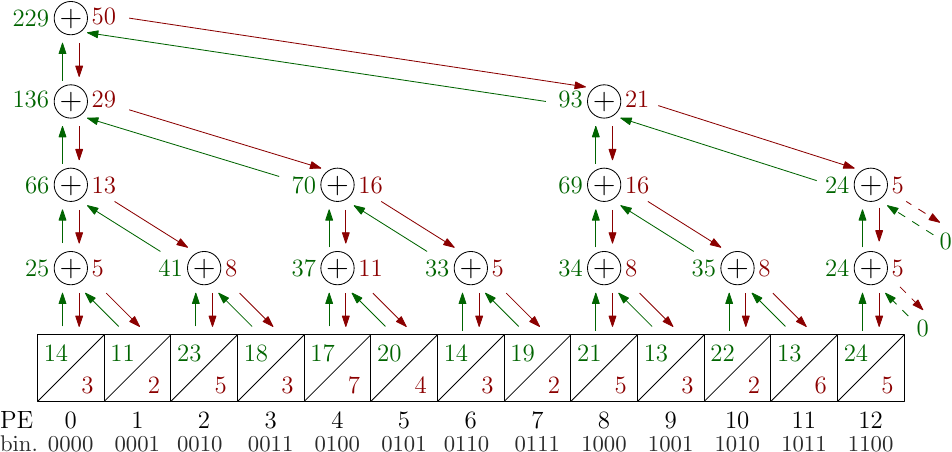}
  \caption{\label{fig:uneven}Assigning $n=50$ samples to $13$ PEs, with $N=229$
    elements.  Element counts ($L$-values, in green, on the nodes' left side)
    are added bottom-up, sample counts ($n'$, in red, right side) are assigned
    top-down.}
\end{figure}

We address this situation by arranging the PEs into a binomial
tree \citep{SulBas77}. Let the PEs be numbered $0..p-1$ now to facilitate binary arithmetic.  If
the binary representation of the PE number $i$ ($\ceil{\log p}$
bits) contains $k$ trailing zeroes, it is connected with PEs
$i+2^j$ for $j\in 0..k-1$ if $i+2^j<p$. The connections for each
value of $j$ form one level of a binary tree, as shown in the example in
Figure~\ref{fig:uneven}.
At level $j\in 0..\ceil{\log p}$ we get (maximal) subtrees spanning PEs
$2^ja..\min(2^ja+2^j-1,p-1)$ for $a=0..p/2^j-1$.
In an upward pass, iterating from $j=0$ upwards,
we compute the sum of the $L$-values in each of these subtrees.
For an inner node let $L_{\ell}$ and $L_r$ denote the partial sums for its left
and right subtrees, respectively.

Now the number of samples in each subtree is computed in a top down
fashion. The root knows that it has to generate $n'=n$ samples.  Other
nodes receive their $n'$ value from their parent.  An inner node uses a
hypergeometric distribution with parameters~$n$, $L_{\ell}$, and
$L_{\ell}+L_r$ to split its $n'$ samples into $n'=n_{\ell}+n_r$.
Then $n_{\ell}$ is used for the next smaller subtree locally, while
$n_r$ is passed to the right child as the number of samples to be
generated there.  To generate independent random values everywhere, the
subtree representing PEs $a..b$ can use this range as an input
for the hash function~$h$ from Algorithm~P.

\subsection{Using True Randomness}\label{ss:true}
Now let us assume that each PE has access to some independent
physical source of truly random values. In this case, we can use the
algorithm from Section~\ref{sss:uneven} since it makes every random
decision only once and explicitly passes the resulting information to
other PEs.

\subsection{Deterministic Results}
For fixed $p$ and $h$, Algorithm~P deterministically and reproducibly
generates the same sample every time, which is important to make
software using the algorithm predictable, reliable, and testable.  If
we even want the result to be independent of $p$, we can use the
load balancing method from Section~\ref{sss:lb}.  In this case, we generate $p'\gg p$ jobs
regardless of the actual number of PEs used and then assign the jobs
to the PEs (possibly even statically, $\ceil{p'/p}$ consecutive jobs for each PE).

\subsection{Sampling with Replacement in Various Spaces}\label{ss:replacement}
Algorithms~\AR and P are easy to adapt to sampling with replacement. The only thing
that changes is that the hypergeometric distribution for the
divide-and-conquer step has to be replaced by a binomial distribution.
Note that this is not restricted to sampling from the one-dimensional
discrete range $1..N$. We can also uniformly sample from continuous or
higher-dimensional sets as long as we can bipartition the space.  For
example, for generating random points in a rectangle we can
subsequently bisect this rectangle into smaller and smaller rectangles
up to some base case. In order to match the size of these base objects
to the number of PEs, it might be useful to generate~$K\gg p$
base objects and to use some kind of load balancing to map base
objects to PEs. This works similar to the load balancing
methods from Section~\ref{sss:lb}.

\subsection{Relation to Bernoulli Sampling}\label{ss:Bernoulli}

We want to point out that Bernoulli sampling and sampling without
replacement are almost equivalent in the sense that they can emulate
each other efficiently. On the one hand, Bernoulli sampling with success probability $\rho$ can be
implemented by sampling without replacement if we can first determine
how many elements $n$ are sampled by Bernoulli sampling.  This number
follows a binomial distribution with parameters $N$ and $\rho$. Then we
can use sampling without replacement to choose the actual elements.
On machines with slow floating point arithmetics,
e.g. microcontrollers, this approach might be faster than generating skip values
from a geometric distribution, which requires evaluating logarithms.

On the other hand,
Algorithm~\ABern \citep{AhrDie85}
generates $n$ samples without replacement by ``repairing'' a Bernoulli sample.
For this paper, it is important that Bernoulli sampling can also be parallelized
in  several ways.  We can independently apply Bernoulli sampling to
subranges of $1..N$. This is the method of choice for distributed
memory machines since it requires no communication. On a shared memory
machine, we can also generate an array of $(1+o(1))\rho N$ independent,
geometrically distributed random deviates and compute
their prefix sums. The values up to $N$ denote the sample.  A
practically important observation is that the operation needed for
this approach has no conditional branches or random memory accesses,
and hence can be implemented on SIMD (single instruction multiple
data) units of modern CPUs or on GPUs.

To parallelize Algorithm~\ABern, we can use
it as base case of Algorithm~P.  We can also use parallel Bernoulli sampling and then
use  Algorithm~P in the repair step.

\section{Implementation Details}\label{s:implementation}

We have implemented algorithms~D, H, R, P, and \ABern using C$++$.%
\footnote{\url{https://github.com/sebalamm/DistributedSampling} and
  \url{https://github.com/lorenzhs/sampling}}
Refer to Table~\ref{tab:algorithms} for a summary of the abbreviations.
We use Spooky
Hash\footnote{\url{http://www.burtleburtle.net/bob/hash/spooky.html}, version 2}
as a hash function which generates seeds for
initializing the Mersenne twister \citep{MatNis98} pseudorandom number generator for
uniform deviates.

\fakepar{Algorithm~D} has been translated literally from the description of \cite{Vit84}.

\fakepar{Algorithm~H} uses hashing with linear probing \citep{Knu98} using a
power of two as table size.  We use two variants for obtaining the
entries of the table and for emptying it.  The default is to record
the positions of inserted elements on a stack. This way, we can
retrieve and reset the table elements without having to consider empty
entries. In turn, this allows us to make the table size $m$
significantly larger than the final number of entries $n$ in order to
speed up table accesses. This does not work when we want to output
table entries in sorted order. Here we omit the stack and explicitly scan the
table at the end. Furthermore, we allocate $n$ additional table
entries to the right so that it becomes unnecessary to wrap around when
an insertion probes beyond the $m$-th table entry. Otherwise, wrapping around
could destroy the globally sorted order between clusters (see Section~\ref{ss:sorted}).

\fakepar{Algorithm~R} uses Algorithm~H as the base case sampler
(\emph{sampleBase} in the pseudocode of Figure~\ref{fig:algorithmR}).  We do this
because Algorithm~H is faster than Algorithm~D for small subproblems
where the hash table fits into cache.  This will always be the case if~$n_0$
is chosen appropriately (we use $n_0=2^{9}$ and $m=2^{12}$).
To generate hypergeometric random deviates, we use the \texttt{stocc}
library\footnote{\url{http://www.agner.org/random/}, version 2014-Jun-14},
which uses a Mersenne twister internally.

\fakepar{Algorithm~P} on Blue~Gene/Q is parallelized using MPICH 1.5 on gcc 4.9.3.
It uses Algorithm~\AR with parameters $n_0=2^8$ and $m=2^{11}$ as local sampling algorithm.

\fakepar{Algorithm~B} uses Algorithm~R for selecting samples to be removed in the repair
step.  Geometric random deviates are generated using the C$++$ standard library classes
\texttt{std::geometric\_distribution} and \texttt{std::mt19937\_64}.

We implemented two further variants of Algorithm~\ABern.  One targets
SIMD parallelism within a single CPU-core. The other uses NVIDIA GPUs.
The CPU-SIMD version performs best when restricting
arithmetics to 32 bits. Therefore we use a smaller maximal universe
size of $N=2^{30}$ there. This version uses the Intel Math Kernel
Library MKL v11.3 \citep{intel-mkl} to generate geometric deviates.  Prefix
sums are computed by a manually tuned routine using SSE2 instructions
through compiler intrinsics.

The GPU version uses CUDA 7.5, the cuRAND
library\footnote{\url{http://docs.nvidia.com/cuda/curand/}, v7.5}
for generating geometric random deviates and the
Thrust library\footnote{\url{https://developer.nvidia.com/thrust}, v1.7.0}
for computing prefix sums. Thus, most of the work
can actually be delegated to libraries tuned by the vendor.
Unfortunately, the repair step, albeit requiring
only sublinear work, is difficult to do on the GPU. Therefore it is
partially delegated to the CPU. There are various ways to accomplish this but
the key point is to do it in a way such that the sample does not need to be
transferred to the CPU.  Our solution first uses a parallel
GPU pass over the sample to count the number $n'$ of prefix sum values
$<N$ (see Section~\ref{s:preliminaries}). Only the single value $n'$ needs to be
transferred to the CPU.  The CPU then uses Algorithm~R to generate $n'-n$
samples from the range $0..n'-1$. These samples are transferred to the GPU which marks
the appropriate positions in the sample array for removal.
Finally, the sample array is compacted using the Thrust function \texttt{copy\_if}.

\section{Experiments}\label{s:experiments}

\begin{figure}
  \centering
  \includegraphics[width=\textwidth]{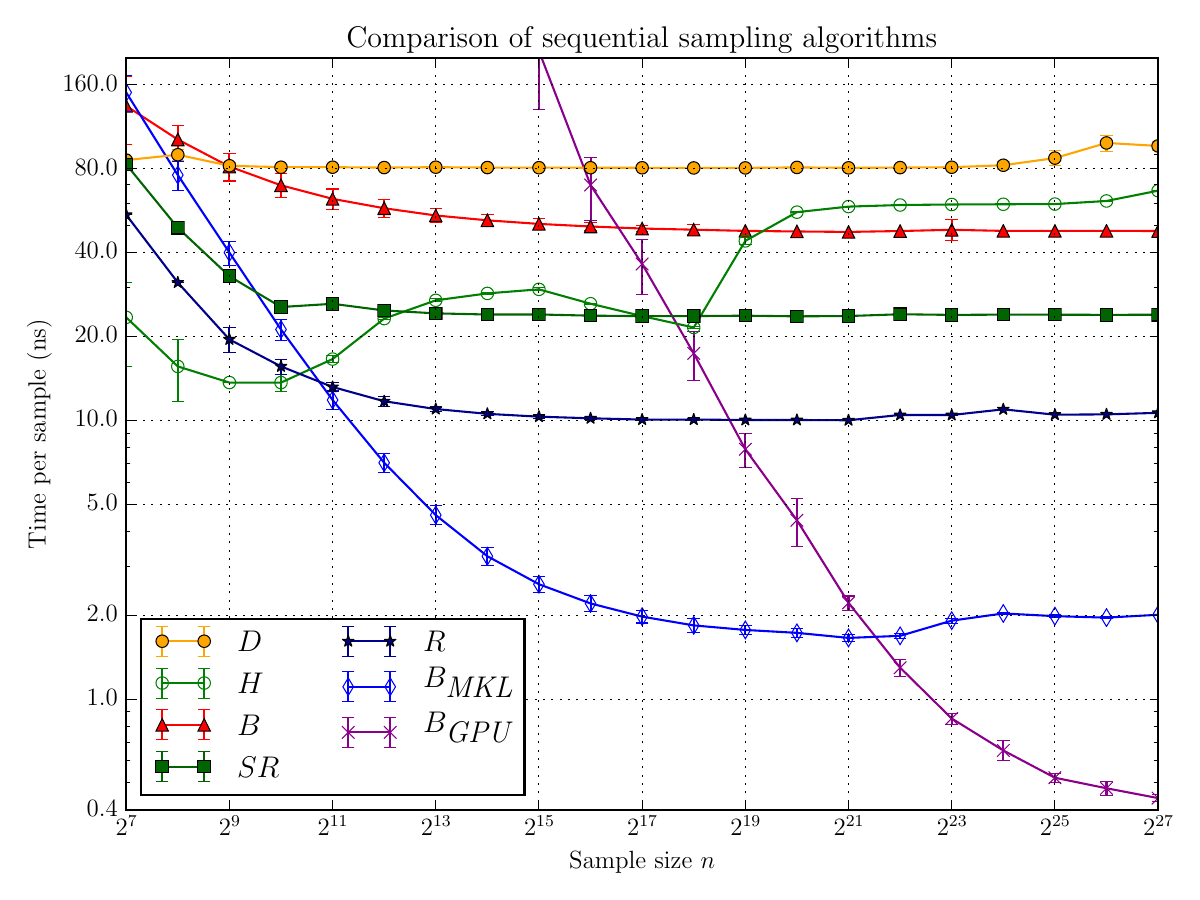}
  \caption{\label{fig:sequentialSpeed}Running time per sample for the sequential
    algorithms H, D, \AR, and \ABern. The bars show the standard deviation.
    The number of repetitions for each algorithm is $2^{30}/n$.
    For Algorithm~\AR, we use $n_0=2^{10}$.\ \ \textit{SR} is Algorithm~\AR with sorted
    output.  $B_\textit{MKL}$ and $B_\textit{GPU}$ are non-portable vectorized
    implementations of Algorithm~\ABern for CPUs using Intel's Math Kernel
    Library (MKL) and NVIDIA GPUs using CUDA, respectively.}
\end{figure}

Figure~\ref{fig:sequentialSpeed} compares the performance of the
sequential Algorithms D, H, R, and \ABern.
Refer to Table~\ref{tab:algorithms} for a summary of the abbreviations.
These experiments were conducted on a single core of a dual-socket
Intel Xeon E5-2670 v3 system with 128\,GiB of DDR4-2133 memory,
running Ubuntu 14.04.  The code was compiled with GNU \texttt{g++} in
version 6.2 using optimization level \texttt{fast} and
\texttt{-march=native}.  We report results for universe size
$N=2^{50}$ and varying $n$.  The number of repetitions
was $2^{30}/n$ to achieve equal work for every $n$.  We see that
Algorithm~H is very fast for small $n$, but its performance degrades
as $n$ grows and the hash table exceeds the cache size.  Our new
Algorithm~\AR is similarly fast for small $n$, but the time per sample
remains constant as $n$ grows.  Thus, it is up to $5$ times faster
than Algorithm~H for very large $n$.  The performance of Algorithm~D
is also independent of~$n$, but worse than Algorithm~\AR by a factor
of $7$. A variant of Algorithm~R (SR) that generates samples online
and in sorted order is still 3.4 times faster than
Algorithm~D.  The portable implementation of Algorithm~B (labeled $B$
in Figure~\ref{fig:sequentialSpeed}) is faster than
Algorithm~D but cannot compete with Algorithm~R.

This picture changes when looking at tuned architecture specific
implementations of Algorithm~B. The CPU version (label $B_\textit{MKL}$) is
up to 6 times faster than Algorithm~R for large $n$. For very large
$n$, the GPU version, $B_\textit{GPU}$, running on an NVIDIA GeForce GTX 980 Ti graphic
card, is yet $4.5$ times faster.  However, it should be noted
that a single core of a Xeon E5-2670 v3
uses much less power than an entire GTX 980 Ti -- the entire Xeon PE
with 12 cores uses about half the power of the graphics card.

\begin{figure}
  \centering
  \includegraphics[width=0.9\textwidth]{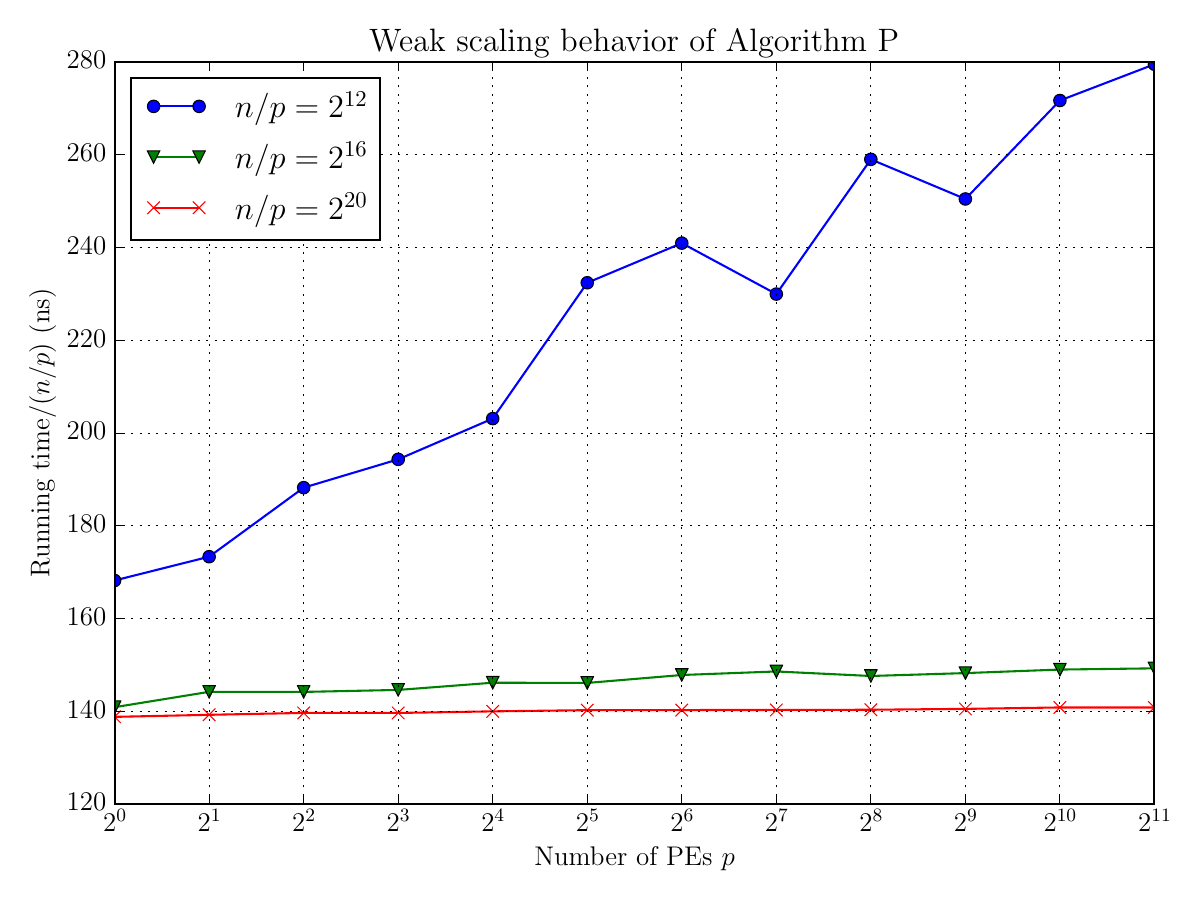}
  \caption{\label{fig:speedup}Running time for generating $n$ samples on $p$
    PEs for different values of $n/p$ using Algorithm~P with
    Algorithm~\AR as local sampler, using $n_0=2^{8}.$
    The number of repetitions for each value of $n/p$ is $2^{30} \cdot p / n$.}
\end{figure}

Our experiments clearly confirm our expectation that offloading
sampling to the GPU for further processing on the CPU is not
worthwhile, as the time for transferring the samples from GPU to CPU memory
(not pictured in Figure~\ref{fig:sequentialSpeed})
dwarfs the time to take the sample -- including transfer, a single core of a modern
CPU can generate the samples equally fast.  However, it also shows
that fast sampling is possible for GPU applications, i.e.~if the
samples are required on the GPU for further processing.

It is also worth looking at the individual components of the running
time of the GPU implementation (we consider the case of
$n=2^{27}$ as an example).  The CPU portion and data transfer account for $13.3\,\%$ of the total running time, while
$86.7\,\%$ are spent on computation on the GPU. This time, in turn, is split up as follows.
Generating geometrically distributed random numbers using
cuRAND takes $25.0\,\%$ of the computation time, and calculating a prefix sum over the
elements using the Thrust function \texttt{inclusive\_scan} takes another $36.7\,\%$.
Counting the number of elements $<N$ with \texttt{count\_if} takes $6.9\,\%$.
While the time for marking the elements selected by the CPU is negligible at $0.2\,\%$,
the following compaction with the Thrust function \texttt{copy\_if} takes
another $31.0\,\%$.

\paragraph*{Algorithm~P}
Figure~\ref{fig:speedup} shows a so-called weak scaling experiment on JUQUEEN, a
distributed memory machine. It shows the running time of Algorithm~P when
keeping local input size $n/p$ constant, measured for different values of this
ratio.  JUQUEEN is an IBM Blue Gene/Q machine, demonstrating the portability of
our code.  We used the maximum number of 16 cores per node for these
experiments.  Performance per core is an order of magnitude lower than on the
Intel~CPU used for our sequential experiments.  A factor of four is more typical
for other applications considering the lower clock frequency, older technology,
and lower number of transistors used.  The remaining factor of $2$--$3$ is mostly
due to the fact that our random number generator, a SIMD-oriented Mersenne
Twister, contains optimizations to make use of the SSE2 units of Intel~CPUs.
However, it does not have similar optimizations for the QPX instructions of Blue
Gene/Q, thus reverting to scalar code.  This is compounded further by the lack
of autovectorization for Blue Gene/Q in gcc.

On the positive side, we see that the code scales almost perfectly
for sufficiently large values of $n/p$.
For the smallest tested value of $n/p$, $4096$, we see a
linear increase in running time with an exponential increase in
$p$. This is consistent with the asymptotic running time of $n/p+\log
p$.

\section{Conclusions}\label{s:conclusion}

We find it surprising that the seemingly trivial problem of random
sampling requires such a diverse set of algorithmic techniques.
Moreover, the features of modern computer architectures entail that no
single approach is universally best. When $n$ is very small,
Algorithm~H is both simple and efficient, but for larger $n$ it
becomes cache-inefficient. This problem can be overcome by using it for the
base case of Algorithm~R. A slight generalization of Algorithm~R
allows for parallelization (Algorithm~P). Since this requires no or almost no
communication, it is suitable for many parallel models of computation, such as shared
memory, distributed memory, or cloud computing. Only some details like load
balancing and adaptation to nonuniform data distribution require
communication.

On the other hand, we see no reason to continue using Algorithm~S.
It is fast (only) if~$N/n$ is a small constant,
but we doubt that it can ever outperform Algorithm~R, which needs at
most half the number of uniform deviates.  In particular, for small $N/n$ we
could use a variant of Algorithm~H for the base case that uses the key
directly to index the table of sampled elements. This avoids the
need for handling collisions between samples.

The main point in favor of Algorithm~D is that it generates the
samples in sorted order and works in an online fashion, i.e., the expected
time between generating samples is constant. With the iterative
version of Algorithm~R described in Section~\ref{ss:sorted} we can
achieve the same effect, but with a more flexible trade-off between
maximum latency between samples and the average cost per sample.  From
that perspective, our divide-and-conquer technique is a generalization
of Algorithm~D that allows faster processing and parallelization.
Actual real time guarantees of deterministic constant time between
subsequent samples seem to be an open problem and neither Algorithm~R
nor~D can offer such guarantees.

Algorithm~B is useful because it allows for vectorization. Hence, on
architectures with fast arithmetics, a tuned version of Algorithm~B
can outperform Algorithm~R. However, this comes at the cost of reduced
portability and that samples cannot be generated in an online
fashion -- it is only after the repair step that we know which samples
survive.

To illustrate the usefulness of fast sampling algorithms, we mention a few
applications. Our algorithms are most useful in settings where fast random access (by index) to the elements is possible, e.g. when dealing with fixed-size records or in main-memory based databases such as SAP HANA
\citep{sikka2013sap}, SAP HANA Vora
\citep{goel2015towards}, or EXASOL (\url{http://www.exasol.com/}).
Generating a random graph in the $G(n,m)$ and $G(n,p)$
model of \cite{ER59rand} is equivalent to sampling from the $n(n-1)/2$
possible edges.
Sampling is performed without replacement for $G(n,m)$ and Bernoulli
sampling is used for $G(n,p)$.
In our library for generating massive graphs \citep{FLS19GraphGen,Lamm17} we successfully use the sampling
algorithms described here.
In the same work, we also use sampling with replacement to generate
point sets for several families of random geometric graphs (proximity based 2D, 3D, Delaunay 2D, 3D, and hyperbolic).
Experiments indicate good performance for up to $2^{18}$ PEs.
Sample sort \citep{BleEtAl91,ABSS15} is a successful
example of a parallel sorting algorithm that splits its input based
on a random sample. With Algorithm~P this is now possible with very low
overhead and without resorting to simplified sampling models, which often
complicate the analysis and reduce sampling quality. This is also an example where scalability matters, i.e., where a bound of $\Oh{n/p+\log p}$ is much better than $\Oh{n/p+p}$.

\subsection*{Acknowledgments}
The authors gratefully acknowledge the Gauss Centre for Supercomputing (GCS) for
providing computing time through the John von Neumann Institute for Computing
(NIC) on the GCS share of the supercomputer JUQUEEN~\citep{juqueen} at J\"ulich
Supercomputing Centre (JSC). GCS is the alliance of the three national
supercomputing centres HLRS (Universit\"at Stuttgart), JSC (Forschungszentrum
J\"ulich), and LRZ (Bayerische Akademie der Wissenschaften), funded by the German
Federal Ministry of Education and Research (BMBF) and the German State
Ministries for Research of Baden-W\"urttemberg (MWK), Bayern (StMWFK) and
Nordrhein-Westfalen (MIWF).

\bibliographystyle{apalike}
\bibliography{diss}

\begin{thebibliography}{}

\bibitem[Ahrens and Dieter, 1985]{AhrDie85}
Ahrens, J.~H. and Dieter, U. (1985).
\newblock Sequential random sampling.
\newblock {\em ACM Trans. Math. Softw.}, 11(2):157--169.

\bibitem[Axtmann et~al., 2015]{ABSS15}
Axtmann, M., Bingmann, T., Sanders, P., and Schulz, C. (2015).
\newblock Practical massively parallel sorting.
\newblock In {\em Proc. 27th ACM Symposium on Parallelism in Algorithms and
  Architectures}, SPAA '15.

\bibitem[Blelloch et~al., 1991]{BleEtAl91}
Blelloch, G.~E., Leiserson, C.~E., Maggs, B.~M., Plaxton, C.~G., Smith, S.~J.,
  and Zagha, M. (1991).
\newblock A comparison of sorting algorithms for the connection machine {CM-2}.
\newblock In {\em Proc. 3rd ACM Symposium on Parallel Algorithms and
  Architectures}, SPAA '91, pages 3--16.

\bibitem[Blumofe et~al., 1995]{BluEtAl95}
Blumofe, R.~D., Joerg, C., Kuszmaul, B.~C., Leiserson, C.~E., Randall, K.~H.,
  and Zhou, Y. (1995).
\newblock Cilk: An efficient multithreaded runtime system.
\newblock In {\em Proc. 5th ACM Symposium on Principles and Practice of
  Parallel Programming}, PPoPP '95, pages 207--216, Santa Barbara, CA, 19--21
  July. ACM New York.

\bibitem[Blumofe and Leiserson, 1999]{BluLei99}
Blumofe, R.~D. and Leiserson, C.~E. (1999).
\newblock Scheduling multithreaded computations by work stealing.
\newblock {\em Journal of the ACM}, 46(5):720--748.

\bibitem[Brown, 2011]{brown6wasted}
Brown, D.~G. (2011).
\newblock How i wasted too long finding a concentration inequality for sums of
  geometric variables.

\bibitem[Erd{\"o}s and R{\'e}nyi, 1959]{ER59rand}
Erd{\"o}s, P. and R{\'e}nyi, A. (1959).
\newblock On random graphs, i.
\newblock {\em Publicationes Mathematicae (Debrecen)}, 6:290--297.

\bibitem[Fan et~al., 1962]{Fan62}
Fan, C.~T., Muller, M.~E., and Rezucha, I. (1962).
\newblock Development of sampling plans by using sequential (item by item)
  selection techniques and digital computers.
\newblock {\em Journal of the American Statistical Association},
  57(298):387--402.

\bibitem[Finkel and Manber, 1987]{FinMan87}
Finkel, R. and Manber, U. (1987).
\newblock {DIB} -- {A} distributed implementation of backtracking.
\newblock {\em ACM Trans. Prog. Lang. and Syst.}, 9(2):235--256.

\bibitem[Funke et~al., 2019]{FLS19GraphGen}
Funke, D., Lamm, S., Meyer, U., Penschuck, M., Sanders, P., Schulz, C., Strash,
  D., and von Looz, M. (2019).
\newblock Communication-free massively distributed graph generation.
\newblock {\em Journal of Parallel and Distributed Computing}, 131:200--217.

\bibitem[Goel et~al., 2015]{goel2015towards}
Goel, A., Pound, J., Auch, N., Bumbulis, P., MacLean, S., F{\"a}rber, F.,
  Gropengiesser, F., Mathis, C., Bodner, T., and Lehner, W. (2015).
\newblock Towards scalable real-time analytics: an architecture for scale-out
  of olxp workloads.
\newblock {\em Proc. VLDB Endowment}, 8(12):1716--1727.

\bibitem[Intel, 2012]{Intel12}
Intel (2012).
\newblock {\em Intel Digital Random Number Generator (DRNG): Software
  Implementation Guide}.
\newblock Intel.
\newblock
  \url{https://software.intel.com/en-us/articles/intel-digital-random-number-generator-drng-software-implementation-guide}.

\bibitem[Intel, 2015]{intel-mkl}
Intel (2015).
\newblock {\em Intel Math Kernel Library v11.3}.
\newblock Intel.
\newblock \url{https://software.intel.com/en-us/mkl-reference-manual-for-c}.

\bibitem[Knuth, 1981]{Knu81}
Knuth, D.~E. (1981).
\newblock {\em The Art of Computer Programming---Seminumerical Algorithms},
  volume~2.
\newblock Addison Wesley, 2nd edition.

\bibitem[Knuth, 1998]{Knu98}
Knuth, D.~E. (1998).
\newblock {\em The Art of Computer Programming---Sorting and Searching},
  volume~3.
\newblock Addison Wesley, 2nd edition.

\bibitem[Lamm, 2017]{Lamm17}
Lamm, S. (2017).
\newblock Communication efficient algorithms for generating massive networks.
\newblock Master's thesis, Karlsruhe Institute of Technology (KIT).

\bibitem[Matsumoto and Nishimura, 1998]{MatNis98}
Matsumoto, M. and Nishimura, T. (1998).
\newblock Mersenne twister: {A} 623-dimensionally equidistributed uniform
  pseudo-random number generator.
\newblock {\em ACM Trans. Model. Comput. Simul.}, 8:3--30.

\bibitem[Meng, 2013]{meng2013scalable}
Meng, X. (2013).
\newblock Scalable simple random sampling and stratified sampling.
\newblock In {\em Int. Conf. on Machine Learning}, ICML '13, pages 531--539.

\bibitem[Sanders, 1996]{San96e}
Sanders, P. (1996).
\newblock {\em Lastverteilungsalgorithmen f{\"u}r parallele Tiefensuche}.
\newblock Dissertation, Universit{\"a}t Karlsruhe.

\bibitem[Sanders, 2002]{San02b}
Sanders, P. (2002).
\newblock Randomized receiver initiated load balancing algorithms for tree
  shaped computations.
\newblock {\em The Computer Journal}, 45(5):561--573.

\bibitem[Sanders et~al., 2013]{SSM13}
Sanders, P., Schlag, S., and M{\"u}ller, I. (2013).
\newblock Communication efficient algorithms for fundamental big data problems.
\newblock In {\em IEEE Int. Conf. on Big Data}, pages 15--23.

\bibitem[Sikka et~al., 2013]{sikka2013sap}
Sikka, V., F{\"a}rber, F., Goel, A., and Lehner, W. (2013).
\newblock Sap hana: The evolution from a modern main-memory data platform to an
  enterprise application platform.
\newblock {\em Proc. VLDB Endowment}, 6(11):1184--1185.

\bibitem[Singler et~al., 2007]{SSP07}
Singler, J., Sanders, P., and Putze, F. (2007).
\newblock {MCSTL}: The multi-core standard template library.
\newblock In {\em Proc. 13th Int. Euro-Par Conf.}, volume 4641 of {\em LNCS},
  pages 682--694. Springer.

\bibitem[Stadlober, 1990]{Stad90hyp}
Stadlober, E. (1990).
\newblock The ratio of uniforms approach for generating discrete random
  variates.
\newblock {\em Journal of Computational and Applied Mathematics},
  31(1):181--189.

\bibitem[Stephan and Docter, 2015]{juqueen}
Stephan, M. and Docter, J. (2015).
\newblock {J}{\"u}lich {S}upercomputing {C}entre. {JUQUEEN}: {IBM} {B}lue
  {G}ene/{Q} {S}upercomputer {S}ystem at the {J}{\"u}lich {S}upercomputing
  {C}entre.
\newblock {\em Journal of large-scale research facilities}, A1:1--5.

\bibitem[Sullivan and Bashkow, 1977]{SulBas77}
Sullivan, H. and Bashkow, T.~R. (1977).
\newblock A large scale, homogeneous, fully distributed parallel machine, i.
\newblock In {\em Proc. 4th Annual Symposium on Computer Architecture}, ISCA
  '77, pages 105--117, New York, NY, USA. ACM.

\bibitem[Vitter, 1984]{Vit84}
Vitter, J.~S. (1984).
\newblock Faster methods for random sampling.
\newblock {\em Commun. ACM}, 27(7):703--718.

\end{thebibliography}

\end{document}